\documentclass[11pt,a4paper]{article}

\usepackage[T1]{fontenc}
\usepackage[utf8]{inputenc}
\usepackage{lmodern}

\usepackage[margin=2.5cm]{geometry}

\usepackage{amsmath}
\usepackage{amssymb}
\usepackage{amsthm}
\usepackage{mathtools}       
\usepackage{listings}        
\usepackage{xcolor}

\usepackage{booktabs}      
\usepackage{microtype}       
\usepackage{tikz}
\usetikzlibrary{positioning, arrows.meta, shapes.geometric}
\usepackage{algorithm}
\usepackage{algpseudocode}

\usepackage[colorlinks=true,
            linkcolor=black,
            citecolor=black,
            urlcolor=blue]{hyperref}

\usepackage{cite}

\newtheorem{theorem}{Theorem}
\newtheorem{lemma}[theorem]{Lemma}
\newtheorem{corollary}[theorem]{Corollary}
\newtheorem{proposition}[theorem]{Proposition}
\newtheorem{definition}[theorem]{Definition}

\newtheorem{example}[theorem]{Example}
\newtheorem{remark}[theorem]{Remark}

\newcommand{\R}{\mathcal{R}}          
\newcommand{\nf}[1]{\mathrm{nf}(#1)}  
\newcommand{\xorfold}{\oplus}          
\newcommand{\scope}[1]{\mathcal{S}_{#1}}
\newcommand{\tostep}{\to}
\newcommand{\tomany}{\to^{*}}
\newcommand{\dist}[1]{\mu_{#1}}   
\newcommand{\bern}{\mathrm{Bern}}     
\newcommand{\unroll}{\mathrm{unroll}} 

\lstdefinelanguage{DEM}{
  morekeywords={error, REPEAT, DETECTOR, OBSERVABLE_INCLUDE},
  sensitive=false,
  morecomment=[l]{\#},
}
\title{Quasilinear Equivalence Checking\\
       for Detector Error Models}
\author{Mathys Rennela\thanks{Unitary Foundation, France}}

\date{\today}

\begin{document}

\maketitle

\begin{abstract}
A Detector Error Model (DEM) is a structured representation of error mechanisms in quantum circuits, which has gained popularity in quantum compilation pipelines for its ability to capture fault-tolerance at a circuit level. 
It lists error mechanisms as instructions targeting detectors and observables, specifying for each physical fault channel the probability that the fault fires, the detectors it triggers, and the observables it flips.

In this paper, we develop an equational theory for DEMs, with its associated categorical semantics. We present a sound, terminating, confluent rewriting system for DEM terms, formulating it as a symmetric monoidal theory (a PROP) over the Giry monad. 
We prove that every DEM term has a unique normal form, which can be computed efficiently in quasilinear time $O(k|E|\log|E|)$, where $|E|$ is the number of instructions and $k$ bounds the size of a target set. This provides a complete set of invariants (via Tanner graphs) for structural DEM equivalence.
We provide the first static decision procedure for DEM equivalence, with rigorous correctness guarantees. It is complete (decides full decoder-equivalence exactly) for non-adaptive quantum error correction (QEC) pipelines, and scales to a sound and applicable decision procedure for partially-adaptive circuits (lattice surgery, distributed QEC, ...) without suffering exponential overhead. We discuss its application to the verification and optimisation of quantum compilers.
\end{abstract}

\section{Background}
\label{sec:background}

A core mechanism of quantum error correction (QEC) is the measurement of error syndromes, which witness the occurrence of physical faults in the quantum hardware. Decoders are (classical) algorithms which guess which error occurred from the observed error syndromes, and decide how to correct it, based on the fault model it was trained on or configured with. Because one cannot directly measure quantum data without collapsing it, one must statistically infer the most likely failure point that explains the triggered detectors.

Stim~\cite{gidney2021stim} is a software library for quantum error correction designed for high-performance simulation of quantum circuits under realistic noise models. It has become a standard tool, and is widely used in quantum error correction software research and development. Stim encodes the fault model associated to a quantum circuit as a \emph{Detector Error Model} (DEM): a flat or structured list of
instructions of the form
\[
  \texttt{error}(p)\; D_{i_1}\cdots D_{i_m}\; L_{j_1}\cdots L_{j_w}
\]
asserting that some physical fault channel fires independently with
probability $p$, triggers detectors $D_{i_1},\ldots,D_{i_m}$, and
flips logical observables $L_{j_1},\ldots,L_{j_w}$.\footnote{The listing of detectors before observables is a presentational convention;
the target list is an unordered multiset over detector and observable indices,
and the order of targets within an instruction carries no semantic content.
}

Stabilisers are operators that leave every state invariant within a given code space, so that measuring their eigenvalues reveals information about the occurrence of computation faults without collapsing the encoded quantum state. Syndrome extraction cycles are quantum circuits which measure stabilisers, and in doing so define a measurement schedule. Such circuits produce measurement outcomes, called syndrome bits, which put together form error syndromes. Detectors are Boolean functions of syndrome bits. They detect faults in the circuit by evaluating to True (1) when a fault triggers its specific combination of syndrome bits.
Observables witness the effect of faults on logical qubits: they evaluate to True when a fault flips the value associated to its logical qubit. We refer the interested reader to Gidney's Stim presentation paper~\cite{gidney2021stim} and Stim's official documentation~\cite{stimdocs} for an introduction to the syntax and semantics of DEMs, and to Gottesman's lecture notes~\cite{gottesman2009qec} for an introduction to the stabiliser formalism.

Typically, a DEM is generated from a quantum circuit annotated with a noise model. It is commonly constructed by Pauli error propagation: for each fault location $l$ (for example, a depolarising channel on a gate), record the error instruction corresponding to the detectors that $l$ triggers and the logical observables that $l$ flips. The resulting DEM encodes the full fault model of the circuit and the noise model given as input, in a compact and decoder-agnostic text file.

With the popularity of Stim, DEMs have emerged as a standard representation in quantum error correction, bridging the gap between noisy quantum circuits and classical decoding algorithms. They decouple error mechanism specification from decoding, providing a standard framework for any decoding library to plug onto. A DEM can be seen as a lossy compression of a noisy circuit, as it captures just what the decoder needs: detectors and logical observables for every source of errors, all in an efficient text format. The format is general enough to capture a wide range of QEC codes, while accurately describing complex noise frameworks.

DEMs are given as input to decoders, which assume that the information provided by the DEM is correct. Moreover, two DEMs can describe the same fault model. This happens for example when a circuit optimisation pass rewrites a circuit in a way that preserves the underlying fault model, or in complex compilation pipelines where multiple quantum circuits (patches) are stitched together (see for example Litinski's game of lattice surgery~\cite{litinski2019lattice}). Deciding whether two DEMs are equivalent, i.e. whether they describe the same fault model, is crucial in the verification that a quantum compilation pipeline preserves fault-tolerance. This equivalence problem cannot be viably tackled by simulation, because the sample complexity required to distinguish two DEMs with high confidence is often exponential in the presence of rare events~\cite{takou2025estimatingdecodingcoherenterrors}. 

In this work, we propose a formal categorical abstraction for DEM equivalence, pivoting away from raw file semantics to provide a complete equational theory of error models. In the same way that Kleene Algebra with Tests (KAT)~\cite{kozen1997kleene} provides verifiable completeness for classical program flows, we aim to provide static, deterministic guarantees for QEC fault models. By formalising DEM extraction as a functor into a discrete symmetric monoidal category, we establish a decision procedure that operates independent of simulation. This categorical structure formalises that error events compose sequentially and combine independently.

We define a term language for DEMs and a notion of scope-local equivalence (Section~\ref{sec:terms}), and present a rewriting system $\R$ on DEM terms representing a symmetric monoidal theory (Section~\ref{sec:rules}). We prove that $\R$ is sound, terminating, and confluent, yielding a decision procedure for scope-local equivalence that runs in quasilinear time $O(k|E|\log|E|)$ via unique complete invariants (Section~\ref{sec:results}). We further show that for QEC pipelines where logical measurements occur at fixed schedule points (making DEMs observable-separated), scope-local equivalence is equivalent to full decoder-equivalence. We then extend this to partially-adaptive circuits via observable-coherence classes (Section~\ref{sec:completeness}). We conclude in Section~\ref{sec:discussion} with directions for future work.

\paragraph{Related work.} 
Verifying the correctness of quantum error correction protocols is an active research topic in formal methods~\cite{carette2019szx,chareton2026hybrid,chen2025verifyingftqc,sundaram2025hoaremeetsheisenberg,huang2026,fang2024symbolic,mei2024eqcheckmodelcounting}.
Several works have studied DEMs as a mathematical formalism for quantum error correction, since Gidney's tool presentation paper for Stim~\cite{gidney2021stim}. 
Such works focus on the expressivity of the formalism~\cite{derks2025designingFTQC}, and on estimation methods~\cite{arms2026demwillow,blumekohout2025estimating,takou2025estimatingdecodingcoherenterrors}, but not on the problem of DEM equivalence. Our use of monadic composition (via the Giry monad~\cite{giry1982categorical}) is inspired by the rich literature on categorical probability theory (see e.g.~\cite{jia2021commutativemonads,affeldt2020trustful,moggi1991notions,Fritz2020}).
The problem of quantum circuit equivalence has been studied (see e.g.~\cite{peham2022equivalenceZX,wille2009,huang2026,yamashita2010fastecqc,mei2024eqcheckmodelcounting,rodatz2026faulttoleranceconstruction}), but the DEM equivalence problem differs for two reasons: two equivalent circuits may not produce the same DEM, and it is unknown whether there is a translation from circuits to DEMs that preserves equivalence; the notion of equivalence for circuits is much more fine-grained than the one for DEMs, which only captures decoder-equivalence.

\section{A term language for detector error models}
\label{sec:terms}

The library Stim~\cite{gidney2021stim} introduced a structured file format for DEMs (.dem), which passes raw fault configurations to decoders. Rather than relying on rigid file-parsing semantics, we elevate DEMs to a formal algebra.

Let $\Sigma$ be the space of all possible DEM states, structured as a \emph{discrete symmetric monoidal category}~\cite{maclane1998categories}. The objects are finite sets of detector and observable indices. A DEM state $\sigma \in \Sigma$ with target set $S$ is a morphism $I \to S$ (a process triggering those targets), and the tensor product $\oplus$ acts as disjoint union: $\sigma_1 \oplus \sigma_2$ is defined if and only if their target sets are disjoint ($S_1 \cap S_2 = \emptyset$). In this categorical framework, generating a DEM (e.g., via Stim's \texttt{detector\_error\_model()}) is simply a \emph{functor}\footnote{A functor is just a structure-preserving map between two mathematical spaces. Here, it guarantees two intuitive physical laws: sequential quantum gates map structurally to sequentially chained error mechanisms, and spatially disjoint quantum operations map transparently to the disjoint ($\oplus$) tensor. Formulating extraction this way ensures that structurally independent noise events in the quantum hardware decompose cleanly and predictably into the DEM term grammar.} $F: \mathbf{Circuit} \to \mathbf{\Sigma}$ from the category of fault-tolerant stabiliser circuits. In other words, generating DEMs is a structured and compositional translation process.

To reason algebraically about these states, we abstract the \texttt{.dem} file format into a term language. The grammar of this language, defined below over countably infinite sets $\mathbb{D}$ (detectors) and $\mathbb{O}$ (observables), captures the right level of abstraction to define DEM equivalence.

\begin{definition}[DEM term grammar]
\label{def:grammar}
A \emph{DEM term} is a well-formed expression $t$ in the following grammar:
\[
  t \;::=\; \varepsilon \;\mid\; \mathtt{error}(p)\;S \;\mid\; t\,;\,t
    \;\mid\; \mathtt{REPEAT}\;N\;\{t\}
\]
where:
\begin{itemize}
  \item $\varepsilon$ is the \emph{empty term};
  \item $;$ is the \emph{sequential composition} operator, with unit $\varepsilon$;
  \item $p \in [0, 0.5)$ is a probability (see Remark~\ref{rem:prob} for this restriction);
  \item $S$ is a finite multiset of \emph{targets} drawn from $\mathbb{D} \cup \mathbb{O}$,
    where elements of $\mathbb{D}$ are called \emph{detectors} and elements of $\mathbb{O}$
    are called \emph{observables};
  \item $N \in \mathbb{N}_{\geq 1}$ is a positive integer, so that $\mathtt{REPEAT}\;N\;\{t\}$ is a \emph{repeat block} with body $t$ and repetition count $N$.
\end{itemize}

The full expression $\mathtt{error}(p)\,S$ is called an \emph{instruction}. Targets carry XOR semantics, so a target appearing an even number of times has no effect. 
\end{definition}

\begin{example}
Consider a simple DEM term with two detectors ($D_0, D_1$) and two observables ($L_0, L_1$).
\begin{lstlisting}[language=DEM]
error(0.1) D0 D1 L1
error(0.2) D1 L0
\end{lstlisting}
is the composition of two instructions: $t = \mathtt{error}(0.1)\{D_0, D_1, L_1\} \,;\, \mathtt{error}(0.2)\{D_1, L_0\}$. 
This DEM states that with 10\% probability a fault triggers $D_0$ and $D_1$ together and flips the logical outcome $L_1$, and independently with 20\% probability another fault triggers $D_1$ and flips the logical outcome $L_0$. If both faults occur simultaneously, $D_1$ is flipped twice and thus cancels out (XOR semantics). 
\end{example}

It is important to note that the full DEM format (.dem) specified by Stim includes annotation instructions (\texttt{detector}, \texttt{logical\_observable}, \texttt{shift\_detectors}) that declare and position detectors but do not introduce error mechanisms. The format also supports decomposition hints via the separator target \texttt{\^{}}, and in the latest versions of Stim (from v1.15), an optional metadata tag on each instruction.
We restrict DEM terms to the \emph{error-mechanism fragment}: the grammar
contains only \texttt{error} instructions and \texttt{REPEAT} blocks.
We assume that detector indices are \emph{absolute}, i.e.\ acting as if all
\texttt{shift\_detectors} offsets have been resolved. Decomposition hints (\texttt{\^{}}) and metadata tags are syntactic sugar, and therefore semantically irrelevant.

\begin{remark}
In Stim, a \texttt{shift\_detectors} instruction inside a repeat block advances the detector offset at each iteration, so the same relative detector index $D_i$ in iteration $k$ and iteration $k{+}1$ resolves to distinct absolute indices. 
Any frontend consuming raw \texttt{.dem} files must first perform a statically-computable resolution of these coordinate shifts (i.e. replace relative indices with absolute, globally unique tracking identifiers) before the absolute-index convention assumed here holds. Once offsets are resolved, the structural \texttt{REPEAT} blocks can be preserved without tracking mutable coordinate state.
\end{remark}

We write $\mathbf{DEM}$ for the set of all DEM terms. In the remainder of the paper, we refer to terms $\mathtt{REPEAT}\;N\;\{t\}$ as \emph{repeat blocks}, and we regard sequential composition $;$ as an associative operator (i.e. $t\,;\,\varepsilon = \varepsilon\,;\,t = t$ and
$(t_1\,;\,t_2)\,;\,t_3 = t_1\,;\,(t_2\,;\,t_3)$ for all terms $t, t_1, t_2, t_3 \in \mathbf{DEM}$). Sequential composition is defined modulo associativity. This means that a scope's instruction sequence is a flat, n-ary list rather than a binary tree. 
We write $I(t)$ and $d(t)$ for the set of instructions and the depth of repeat blocks in a term $t$, respectively. A term is said to be \emph{flat} if it contains no repeat blocks, i.e. if $d(t) = 0$.

\begin{definition}[Scope]
\label{def:term-properties}
The scope of a term $t$ is either the top level of a DEM term, or the body of a repeat block. The set of scopes of a term $t$ is denoted $\scope{t}$, and consists of the top-level scope together with one scope for each occurrence of a $\mathtt{REPEAT}$ block in $t$. Scopes are identified by their position (the path from the root) in the abstract syntax tree of $t$, not by the pair $(N, t')$; distinct $\mathtt{REPEAT}$ blocks with the same repetition count and syntactically identical body give rise to distinct scopes.
\end{definition}

\begin{definition}[Scope distribution]
\label{def:scope-dist}
The \emph{effective target set} $\widehat{S}$ of a target multiset $S$ is
$\widehat{S} = \{\,x \in \mathbb{D} \cup \mathbb{O} : x \text{ appears an odd number of times in } S\,\}$.
The \emph{XOR-fold} operation on probabilities is defined as $p \xorfold q = p + q - 2pq$.
The \emph{XOR-convolution} of two probability measures $\mu, \nu$ on
$\{0,1\}^{\mathbb{D}} \times \{0,1\}^{\mathbb{O}}$ is
$(\mu \circledast \nu)(s, \ell) = \sum_{(s',\ell')} \mu(s',\ell') \cdot \nu(s \oplus_2 s',\, \ell \oplus_2 \ell')$,
where $\oplus_2$ denotes bitwise addition modulo 2 on Boolean vectors.

The \emph{scope distribution} $\dist{t}$ of a term $t$
is the probability measure on $\{0,1\}^{\mathbb{D}} \times \{0,1\}^{\mathbb{O}}$
defined by induction on $t$, where :
\begin{align*}
  \dist{\varepsilon}
    &\;=\; \delta_{(\mathbf{0},\mathbf{0})} \\
  \dist{\mathtt{error}(p)\,S}
    &\;=\; (1-p)\,\delta_{(\mathbf{0},\mathbf{0})}
           \;+\; p\,\delta_{\bigl(\chi_{\widehat{S}\cap\mathbb{D}},\;\chi_{\widehat{S}\cap\mathbb{O}}\bigr)} \\
  \dist{t_1 \,;\, t_2}
    &\;=\; \dist{t_1} \circledast \dist{t_2} \\
  \dist{\mathtt{REPEAT}\,N\,\{t'\}}
    &\;=\; \delta_{(\mathbf{0},\mathbf{0})}
\end{align*}
where $\mathbf{0}$ is the all-zero vector, $\chi_A$ denotes the characteristic vector of a set $A$, and $\delta_x$ denotes the Dirac point mass concentrated at $x$.

\end{definition}

The scope distribution of a repeat block reflects that it contributes no error events at the current scope; its body $t'$ is a separate scope with its own distribution. The semantics of $\mathtt{REPEAT}\,N\,\{t'\}$ is the $N$-fold independent product of the fault distribution defined by body $t'$: each iteration samples faults independently from the same distribution as $t'$.

\begin{definition}[Unrolling]
\label{def:unrolling}
The unrolled term $\unroll(t)$ of a DEM term $t$ is the term obtained by replacing each $\mathtt{REPEAT}\;N\;\{t'\}$ in $t$ with $N$ sequential copies of $t'$. It is defined by induction on $t$:
\begin{align*}
  \unroll(\varepsilon) &= \varepsilon \\
  \unroll(\mathtt{error}(p)\,S) &= \mathtt{error}(p)\,S \\
  \unroll(t_1 \,;\, t_2) &= \unroll(t_1) \,;\, \unroll(t_2) \\
  \unroll(\mathtt{REPEAT}\,N\,\{t'\}) &= \underbrace{\unroll(t')\,;\,\cdots\,;\,\unroll(t')}_{N\text{ times}}
\end{align*}
\end{definition}

The notion of scope is crucial to our formalism, because the notion of equivalence we are interested in this work is not global, but local to each scope. Each rewriting rule is \emph{scope-local}: it can only affect instructions within the same scope. This is a crucial restriction, since it allows to capture the right level of abstraction for DEM equivalence. Indeed, instructions in different repeat blocks are not necessarily independent. Detector error models define a measurement schedule for the underlying circuit, so that detectors in different iterations of a repeat block refer to measurements happening at different times, and thus may correspond to different fault events. For this reason, a repeat block $\texttt{REPEAT}\;$N$\;\{$t$\}$ is not syntactic sugar for $N$ sequential copies of $t$, and if we were to allow rewriting across scope boundaries, our notion of equivalence would be too coarse, as it would identify equivalence of terms which are not actually decoder-equivalent (i.e. which do not describe the same fault model).

\begin{definition}[Decoder-equivalence]
\label{def:decoder-equivalence}
Two DEM terms $t_1, t_2 \in \mathbf{DEM}$ are \emph{decoder-equivalent} if they induce the same joint probability distribution over detector outcomes $\{0,1\}^{\mathbb{D}}$ and observable outcomes $\{0,1\}^{\mathbb{O}}$, where the distribution is computed by replacing each repeat block of count $N$ with $N$ sequential independent copies of its body, i.e.
\[
t_1 \approx t_2 \quad\text{if}\quad \dist{\unroll(t_1)} = \dist{\unroll(t_2)},
\]
\end{definition}

The focus on scope-local relations in this work requires further explanation. We could have chosen to define a global notion of equivalence, that allows to rewrite across scope boundaries. Deciding decoder-equivalence (Definition~\ref{def:decoder-equivalence}) requires unrolling all repeat blocks, making the problem linear in the repetition count of a block. Right now, quantum error correction practitioners rely on statistical equivalence: they sample from two DEMs and use a Monte Carlo method to estimate the total variation distance between the two distributions. This procedure only gives a probabilistic guarantee. A static procedure for scope-local equivalence is decidable and efficient, and respects compositionality.

\section{Rewriting detector error models}
\label{sec:rules}

\begin{figure}[t]
\centering
\begin{align*}
  \mathtt{error}(p_1)\;S \;; \mathtt{error}(p_2)\;S
    &\;\tostep\; \mathtt{error}(p_1 \xorfold p_2)\;S
    &&\text{(R1: XOR fusion)}\\[6pt]
  \mathtt{error}(0)\;S
    &\;\tostep\; \varepsilon
    &&\text{(R2a: zero probability)}\\[6pt]
  \mathtt{error}(p)\;\emptyset
    &\;\tostep\; \varepsilon
    &&\text{(R2b: empty targets)}\\[6pt]
  \varepsilon \;; \; t
    &\;\tostep\; t
    &&\text{(R2c: left unit)}\\[6pt]
  t \;; \; \varepsilon
    &\;\tostep\; t
    &&\text{(R2d: right unit)}\\[6pt]
  \mathtt{error}(p)\;(S \uplus \{L_j, L_j\})
    &\;\tostep\; \mathtt{error}(p)\;S
    &&\text{(R3a: observable cancellation)}\\[6pt]
  \mathtt{error}(p)\;(S \uplus \{D_i, D_i\})
    &\;\tostep\; \mathtt{error}(p)\;S
    &&\text{(R3b: detector cancellation)}\\[6pt]
  \mathtt{error}(p_1) S_1 \;\;;\;\; \mathtt{error}(p_2) S_2
    &\leftrightarrow \mathtt{error}(p_2) S_2 \;\;;\;\; \mathtt{error}(p_1) S_1
    &&\text{(R4: commutativity)}
\end{align*}
\[
  \frac{\displaystyle t \;\tostep\; t'}
       {\mathtt{REPEAT}\;N\;\{t\} \;\tostep\; \mathtt{REPEAT}\;N\;\{t'\}}
  \qquad\text{(R5: repeat body normalisation)}
\]
\caption{Rewriting rules of $\mathcal{R}$. Rules R1--R4 are
  scope-local. Terms are modulo same-scope permutation.
  Rules R2c--R2d eliminate vacuous units introduced by R2a--R2b.
  Rules R3a--R3b eliminate duplicate targets ($\uplus$ is the multiset union).
  Rule R5 is the congruence rule for \texttt{REPEAT} blocks.}
\label{fig:rules}
\end{figure}

We define a rewriting system $\R$ on DEM terms, which presents a symmetric monoidal theory (a PROP~\cite{maclane1965categorical}) for structural DEM equivalence. A PROP can be understood as the formal syntax for \emph{string diagrams}: networks of boxes (error instructions) connected by wires (targets: detectors and observables). The rewriting rules are topological equations that allow us to safely slide, merge, and simplify these boxes, without altering the overarching probabilistic process. In doing so, equational changes become simple graph rewriting operations.

Each rule is scope-local (only affects instructions within the same scope). We define the relation $\tostep$ on $\mathbf{DEM}$ as the smallest relation such that $t_1 \tostep t_2$ if $t_2$ can be obtained from $t_1$ by applying a rewriting rule, and we write $\tomany$ for its reflexive-transitive closure. We write $\nf{t}$ for the normal form of a term $t$ under $\R$, defined as the unique term such that $t \tomany \nf{t}$ and there is no term $t'$ such that $\nf{t} \tostep t'$. The rewriting rules are given in Figure~\ref{fig:rules}. Each rewriting rule is either intra-instruction (it inspects a single instruction and rewrites it), or inter-instruction (it inspects a pair of instructions and rewrites them together). We now define the notion of scope-local equivalence.

\begin{definition}[Term skeleton]
\label{def:skeleton}
The \emph{skeleton} $\rho(t)$ of a term $t$ is the abstract syntax tree of $t$
restricted to \texttt{REPEAT} nodes: it is a finite ordered forest where each
node is labelled by the repetition count $N \in \mathbb{N}_{\geq 1}$  of a \texttt{REPEAT} block, and the children of a node are the \texttt{REPEAT} blocks directly nested inside its body.

Formally, the skeleton $\rho(t)$ is defined by induction on terms as follows:
\begin{itemize}
  \item $\rho(\varepsilon) = \rho(\mathtt{error}(p)\,S) = []$ (the empty sequence);
  \item $\rho(t_1 \,;\, t_2) = \rho(t_1) \cdot \rho(t_2)$ (concatenation of sequences);
  \item $\rho(\mathtt{REPEAT}\,N\,\{t'\}) = [\,\mathrm{node}(N,\,\rho(t'))\,]$ (a singleton sequence containing one tree with root label $N$ and children $\rho(t')$).
\end{itemize}
We say that two terms $t_1$ and $t_2$ are \emph{skeleton-compatible} if $\rho(t_1) = \rho(t_2)$.
\end{definition}

\begin{definition}[Scope-local equivalence]
\label{def:equivalence}
Two skeleton-compatible terms $t_1, t_2 \in \mathbf{DEM}$ (i.e.\ $\rho(t_1) = \rho(t_2)$)
are \emph{scope-locally equivalent}, written $t_1 \sim t_2$, if
\[
  \dist{t_1|_{\mathcal{S}}} = \dist{t_2|_{\mathcal{S}}}
\]
for every scope position $\mathcal{S}$ in their common skeleton, where $t|_{\mathcal{S}}$
is $t$ itself for the top-level scope, and the body $t'$ of the corresponding
$\mathtt{REPEAT}\,N\,\{t'\}$ subterm of $t$ for each nested scope.
\end{definition}

Instructions carry a XOR semantics: each element may appear at most once. This is enforced by Rules R3a--R3b, which eliminate duplicate targets. Rule R2 eliminates trivial instructions (no targets or zero probability).
Rule R5 is a congruence rule, which permits rewriting inside the body of a repeat block, but does not allow rewrites beyond scope boundaries.

Rule R1 fuses same-target instructions by XORing their probabilities. In categorical terms, it arises naturally as \emph{Kleisli composition} in the \emph{Giry monad}~\cite{giry1982categorical} restricted to $\mathbb{F}_2$.\footnote{A physical error mechanism can be seen as a function that outputs a probability distribution over states, and the Giry monad provides the mathematical framework to compose such functions. "Kleisli composition" over $\mathbb{F}_2$ means computing the net probability that two independent coin flips result in an odd number of triggers (i.e. an XOR flip).} Because two probabilistic faults that act on the exact same targets form independent Bernoulli events\footnote{We write $\bern(p)$ for the Bernoulli distribution with success probability $p$.}, their composition operates on the commutative monoid $([0,1], \xorfold)$ with unit element $0$. 

\begin{remark}
\label{rem:prob}
While Stim's syntax permits error probabilities greater than $0.5$, rates $p > 0.5$ represent deterministic flips or anti-correlated noise that fall outside standard stochastic QEC noise models. Thus, we restrict valid DEM terms in our calculus to the physically meaningful regime $p < 0.5$. The boundary $p = 0.5$ represents maximal uncertainty (perfectly uniform outcomes) which mathematically masks other independent events and would collapse the Walsh-Hadamard transform to zero on some characters, violating unique normal forms. The XOR-fold operation is intrinsically and strictly closed on $[0, 0.5)$ (via the multiplicativity of $f(p) = 1-2p$). Because all reductions over $\mathbb{F}_2$ remain confined to this interval, this restriction ensures canonical uniqueness without requiring an explicit normalisation rule.
\end{remark}

\begin{lemma}[XOR-fold properties]
\label{lem:xorfold}
The XOR-fold operation $\xorfold$ is associative, commutative, and has unit element $0$.
\end{lemma}

\begin{proof}
Commutativity follows from the definition of $\xorfold$. Associativity follows from direct computation, since $p \xorfold q = \frac{1-f(p)f(q)}{2}$, where $f(p) = 1 - 2p$. The unit element is $0$, since $p \xorfold 0 = p$ for all $p \in [0,1]$.  
\end{proof}

The order of instructions within the same scope does not matter, as all error events within the same scope are independent, and thus commute. Rule R4 captures the commutativity of same-scope instructions. By introducing this rule, we ensure that term equivalences are modulo permutation within the same scope. This is necessary to ensure that the normal form of a term is unique.

\begin{example}
\label{ex:normalform}
Consider the two DEM terms $t_1$ and $t_2$ below.

\noindent
\begin{minipage}[t]{0.48\linewidth}
\begin{lstlisting}[caption={Term $t_1$}]
error(0.1) D0 D1
error(0.2) D0 D1
error(0.3) D2 L0
repeat 5 {
  error(0.05) D3 D3 D4
  error(0.1) D5 L1
}
\end{lstlisting}
\end{minipage}
\hfill
\begin{minipage}[t]{0.48\linewidth}
\begin{lstlisting}[caption={Term $t_2$}]
error(0.3) D2 L0
error(0.26) D0 D1
repeat 5 {
  error(0.1) D5 L1
  error(0.05) D4
}
\end{lstlisting}
\end{minipage}

\smallskip
\noindent Normalising $t_1$: in the top-level scope, the two instructions with target set $\{D_0, D_1\}$ are fused by R1:
\[
  \mathtt{error}(0.1)\,D_0 D_1 \;;\; \mathtt{error}(0.2)\,D_0 D_1
  \;\tostep\;
  \mathtt{error}(0.1 \xorfold 0.2)\,D_0 D_1
  = \mathtt{error}(0.26)\,D_0 D_1
\]
since $0.1 \xorfold 0.2 = 0.1 + 0.2 - 2 \cdot 0.1 \cdot 0.2 = 0.26$.
In the repeat body, $D_3$ appears twice in the first instruction and is cancelled by R3b:
\[
  \mathtt{error}(0.05)\,D_3 D_3 D_4 \;\tostep\; \mathtt{error}(0.05)\,D_4.
\]
Term $t_2$ is already in normal form. Both terms share the normal form (modulo same-scope permutation):
\begin{lstlisting}[caption={Common normal form $\nf{t_1} = \nf{t_2}$}]
error(0.26) D0 D1
error(0.3) D2 L0
repeat 5 {
  error(0.05) D4
  error(0.1) D5 L1
}
\end{lstlisting}
The decision procedure of Corollary~\ref{cor:decision} concludes $t_1 \sim t_2$ by comparing these normal forms.
\end{example}

\section{Normalising detector error models}
\label{sec:results}

Having defined a grammar for DEM terms and rewriting rules for such terms, we explore the structure of DEMs by proving several properties: 
soundness (rewriting rules preserve meaning), termination (no infinite chains of rewrites), confluence (the rewriting system reaches a unique normal form).
This allows us to define normal forms for DEMs, and a decision procedure for scope-local equivalence. Note that all results in this section are proved modulo same-scope permutation, since R4 (commutativity) is a congruence rule. In other words, our rewriting system preserves the meaning of DEMs, since we only simplify trivial instruction, commute independent error events, and fuse error events acting on the same targets. This observation is formalised in the soundness lemma below.

\begin{lemma}[Soundness]
\label{lem:soundness}
For all terms $t_1, t_2 \in \mathbf{DEM}$, if $t_1 \tostep t_2$, then $t_1 \sim t_2$.
\end{lemma}

\begin{proof}
The proof proceeds by induction on the maximum nesting depth $d(t)$ of \texttt{REPEAT} blocks. All rewriting rules preserve term skeletons, so it suffices to check that they preserve scope distributions. For each rewriting rule, we verify that $\dist{t_1|_{\mathcal{S}}} = \dist{t_2|_{\mathcal{S}}}$ for every scope position $\mathcal{S}$. 

Consider the reduction ($\mathtt{error}(p_1)\,S \;;\; \mathtt{error}(p_2)\,S \;\tostep\; \mathtt{error}(p_1 \xorfold p_2)\,S$) by Rule R1.
Both instructions have identical literal target multisets $S$, which produce the same effective target set $\widehat{S}$ (Definition~\ref{def:scope-dist}). Define $v = (\chi_{\widehat{S}\cap\mathbb{D}},\chi_{\widehat{S}\cap\mathbb{O}})$. Expanding the definition of scope distribution, we observe that 
\begin{align*}
  \dist{\mathtt{error}(p_1)\,S \;;\; \mathtt{error}(p_2)\,S}
  &= \bigl((1-p_1)\delta_{(\mathbf{0},\mathbf{0})} + p_1\,\delta_v\bigr) \circledast \bigl((1-p_2)\delta_{(\mathbf{0},\mathbf{0})} + p_2\,\delta_v\bigr) \\
  &=(1-(p_1\xorfold p_2))\,\delta_{(\mathbf{0},\mathbf{0})} + (p_1 \xorfold p_2)\,\delta_v
= \dist{\mathtt{error}(p_1 \xorfold p_2)\,S} \qquad (v \xorfold v = \mathbf{0}),
\end{align*}
which is exactly the XOR of two independent Bernoulli random variables.

For Rule R2a, $\dist{\mathtt{error}(0)\,S} = (1-0)\,\delta_{(\mathbf{0},\mathbf{0})} + 0\cdot\delta_v = \delta_{(\mathbf{0},\mathbf{0})} = \dist{\varepsilon}$. For Rule R2b, since $\widehat{\emptyset}=\emptyset$ (by Definition~\ref{def:scope-dist}, the effective set of an empty multiset is empty), $\dist{\mathtt{error}(p)\,\emptyset} = (1-p)\,\delta_{(\mathbf{0},\mathbf{0})} + p\,\delta_{(\mathbf{0},\mathbf{0})} = \delta_{(\mathbf{0},\mathbf{0})} = \dist{\varepsilon}$. For Rules R2c and R2d, $\dist{\varepsilon\,;\,t} = \delta_{(\mathbf{0},\mathbf{0})} \circledast \dist{t} = \dist{t}$ and $\dist{t\,;\,\varepsilon} = \dist{t} \circledast \delta_{(\mathbf{0},\mathbf{0})} = \dist{t}$ since $\delta_{(\mathbf{0},\mathbf{0})}$ is the identity for $\circledast$. For Rules R3a and R3b, removing two copies of the same target does not change the effective target set, so that $\dist{\mathtt{error}(p)\,(S \cup \{x,x\})} = \dist{\mathtt{error}(p)\,S}$. For Rule R4, $\dist{\mathtt{error}(p_1)\,S_1 \;;\; \mathtt{error}(p_2)\,S_2} = \dist{\mathtt{error}(p_1)\,S_1} \circledast \dist{\mathtt{error}(p_2)\,S_2} = \dist{\mathtt{error}(p_2)\,S_2} \circledast \dist{\mathtt{error}(p_1)\,S_1} = \dist{\mathtt{error}(p_2)\,S_2 \;;\; \mathtt{error}(p_1)\,S_1}$, since $\circledast$ is commutative.

For Rule R5, if $t \tostep t'$, by the induction hypothesis (applied to the inner step $t \tostep t'$ with $d(t) \le d(\mathtt{REPEAT}\,N\,\{t\})$), $\dist{t} = \dist{t'}$ at the body scope. At the outer scope both blocks contribute $\delta_{(\mathbf{0},\mathbf{0})}$ by definition, thus $\dist{\mathtt{REPEAT}\,N\,\{t\}} = \delta_{(\mathbf{0},\mathbf{0})} = \dist{\mathtt{REPEAT}\,N\,\{t'\}}$ and the equivalence of the blocks holds natively.
\end{proof}

Then, we observe that each rewriting rule (but R4) strictly decreases a lexicographic weight on terms, so there is no infinite sequence of rewrites. 

\begin{lemma}[Termination]
\label{lem:termination}
There is no infinite sequence of terms $t_1, t_2, \ldots$ such that $t_1 \tostep t_2 \tostep t_3 \tostep \cdots$.
\end{lemma}

\begin{proof}
Assign to each term $t$ a weight $w(t) = (C(t), T(t))$ ordered lexicographically, where $T(t)$ is the total number of targets across all error instructions in $t$ (at all scopes), and $C(t)$ is the number of nodes in the abstract syntax tree of $t$ defined by induction as follows:
\begin{align*}
  C(\varepsilon) &= 1, &
  C(\mathtt{error}(p)\,S) &= 2, \\
  C(t_1\,;\,t_2) &= C(t_1)+C(t_2)+1, &
  C(\mathtt{REPEAT}\,N\,\{t'\}) &= 1+C(t'),
\end{align*}
Rule R4 is a quotient and generates no directed reductions. Every other rule strictly decreases $w$ in lexicographic order:
\begin{itemize}
  \item R1 collapses two instruction nodes and one composition node ($C = 5$) into one instruction node ($C = 2$).
  \item R2a--R2b each replace an instruction node ($C = 2$) with $\varepsilon$ ($C = 1$).
  \item R2c--R2d remove a composition node and an $\varepsilon$ node, decreasing $C$ by $2$.
  \item R3a--R3b leave $C$ unchanged and decrease $T$ by $2$.
  \item R5 (congruence): if any step R1--R3 is applied inside the body $t$, the reduct $t'$ has strictly smaller weight $w(t') < w(t)$. Consequently $w(\mathtt{REPEAT}\,N\,\{t'\}) < w(\mathtt{REPEAT}\,N\,\{t\})$, preserving the strict decrease.
\end{itemize}
Since $w(t) \in \mathbb{N}^2$ strictly decreases at each step, no infinite rewriting sequence exists.
\end{proof}

\begin{lemma}[Local confluence]
\label{lem:local-confluence}
For every term $t$ and every pair of one-step reduction $t \tostep u$ and $t \tostep v$, there exists 
$w$ such that $u \tomany w$ and $v \tomany w$.
\end{lemma}

\begin{proof}
Local confluence is the property that every critical pair of one-step reductions is joinable. Rule R4 is a quotient (no directed overlaps), so we only need to check critical pairs generated by Rules R1--R3 and R5. We proceed by induction on the REPEAT nesting depth $d(t)$. 

\textit{Base case} ($d(t) = 0$). Rule R5 does not apply (no depth means no repeat blocks). Therefore one only needs to check critical pairs generated by Rules R1--R3. Since R1 is the only two-instruction rule, every critical pair on Rules R1--R3 falls into one of three shapes:
\begin{itemize}
\item two R1 reductions (joinable by associativity of $\xorfold$, since both orders of fusion yield $\mathtt{error}(p_1\xorfold p_2\xorfold p_3)\,S$);
\item R1 with a one-instruction rule reducing one of its inputs (joinable because $0\xorfold p = p$ eliminates the R2a case, $\widehat{\emptyset}=\emptyset$ collapses both paths to $\varepsilon$ in the R2b case, and $\widehat{S\cup\{x,x\}}=\widehat{S}$ makes R3 commute with R1);
\item two overlapping R3 reductions (joinable since successive target cancellations commute).
\end{itemize}

\textit{Inductive step} ($d(t) > 0$). Assume local confluence holds for all terms of nesting depth $< d(t)$. Any repeat block $\mathtt{REPEAT}\,N\,\{t'\}$ in $t$ has body $t'$ with $d(t') < d(t)$. By the inductive hypothesis, $t'$ is locally confluent, so any two R5 reductions on the same body converge. For critical pairs between R5 and Rules R1--R3: R5 applies exclusively inside the body while R1--R3 apply at the outer scope, so there is no overlap and the pair is trivially joinable. The base-case argument covers all outer-scope critical pairs.
\end{proof}

By Newman's lemma~\cite{newman1942, baader1998term}, termination and local confluence imply confluence.

\begin{proposition}[Confluence]
\label{prop:confluence}
The rewriting system $\R$ is confluent: if $t \tomany t_1$ and $t \tomany t_2$, then there exists a term $t'$ such that $t_1 \tomany t'$ and $t_2 \tomany t'$.
\end{proposition}

The existence of a normal form for every term follows from termination, and its uniqueness follows from confluence. We can therefore state the normal form theorem for DEM terms.

\begin{theorem}[Normal form]
\label{thm:normalform}
Every term $t \in \mathbf{DEM}$ has a unique normal form $\nf{t}$ under $\R$ (up to same-scope permutation of instructions), i.e. $t \tomany \nf{t}$, and there is no term $t'$ such that $\nf{t} \tostep t'$. 
\end{theorem}

The normal form theorem has a clean algebraic reading. Within a single scope, Rules R1--R3 identify two terms that have the same set of fused instructions (up to order); Rule R4 makes order irrelevant. The normal form of a flat term (no \texttt{REPEAT} blocks) is therefore a \emph{finite set} of error instructions with pairwise distinct target multisets, forming a commutative monoid with sequential composition $;$ as operation and $\varepsilon$ as unit. By the normal form theorem, $(\mathbf{DEM}_0 / {\sim},\, ;\,, \varepsilon)$ is the \emph{free commutative monoid} on canonically-fused error instructions (modulo rule R1, such that generators do not share target sets), where $\mathbf{DEM}_0$ denotes the flat fragment. The full term language is then a tree of such monoids, one per scope, since each \texttt{REPEAT} block introduces a fresh scope, whose body is independently a free commutative monoid of fused instructions.

This structure admits a direct graphical interpretation in the language of categorical string diagrams.

\begin{definition}[Tanner Graph]
\label{def:tanner}
The \emph{Tanner graph}~\cite{tanner1981} $G(t)$ of a flat DEM term $t$ is a bipartite graph $(V_L, V_R, E)$ where the left nodes $V_L = \mathbb{D} \cup \mathbb{O}$ are the targets, the right nodes $V_R$ are the error instructions, and an edge connects $v \in V_L$ to $e \in V_R$ if the target $v$ is triggered/flipped by the instruction $e$. Each instruction node $e \in V_R$ is labelled with its Bernoulli probability $p_e \in [0, 0.5)$.
\end{definition}

In the DEM PROP, a term's Tanner graph $G(t)$ is simply its string diagram, a visual representation of the term with its structure and dependencies. The rewriting rules of $\R$ preserve this diagram modulo the topological cancellation of duplicate edges (Rules R3a--R3b) and the fusion of instruction nodes with identical neighborhoods (Rule R1). Since probabilities are bounded in $[0,0.5)$, the normal form of the graph is topologically and quantitatively unique. It is worth noting that a Tanner graph's bipartite graph structure together with its probability labels encode the scope distribution uniquely, and conversely two terms with the same scope distribution have isomorphic Tanner graphs.

  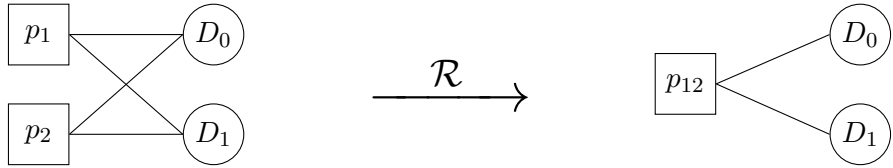
\begin{figure}[ht]
  \centering
  \begin{tikzpicture}[
      target/.style={circle, draw, minimum size=8mm, inner sep=1pt},
      error/.style={rectangle, draw, minimum size=8mm, inner sep=2pt},
      >=Stealth
  ]
      \node (l1) [target] {$D_0$};
      \node (l2) [target, below=0.5cm of l1] {$D_1$};
      \node (e1) [error, left=1.5cm of l1] {$p_1$};
      \node (e2) [error, left=1.5cm of l2] {$p_2$};
      \draw (e1.east) -- (l1.west);
      \draw (e1.east) -- (l2.west);
      \draw (e2.east) -- (l1.west);
      \draw (e2.east) -- (l2.west);
      
      \node (arrow) [right=1.5cm of l1, yshift=-0.65cm] {\huge $\xrightarrow{\quad\mathcal{R}\quad}$};
      
      \node (e1_f) [error, right=1.5cm of arrow] {$p_{12}$};
      \node (l1_f) [target, right=1.5cm of e1_f, yshift=0.65cm] {$D_0$};
      \node (l2_f) [target, below=0.5cm of l1_f] {$D_1$};
      
      \draw (e1_f.east) -- (l1_f.west);
      \draw (e1_f.east) -- (l2_f.west);
  \end{tikzpicture}
  \caption{Tanner graph reduction using rewriting rules $\mathcal{R}$. Error mechanisms with identical targets are fused together algebraically via $p_{12} = p_1 \xorfold p_2$ (Rule R1).}
  \label{fig:tanner_reduction}
  \end{figure}
\begin{theorem}[Complete Set of Invariants]
\label{thm:invariants}
For any flat term $t \in \mathbf{DEM}_0$, the Tanner graph $G(\nf{t})$ of its normal form coupled with its probability assignment forms a complete set of invariants for the equational theory. That is, $t_1 \sim t_2$ if and only if their normal-form Tanner graphs $G(\nf{t_1})$ and $G(\nf{t_2})$ are isomorphic via a bipartite graph isomorphism that acts as the identity on the target nodes $V_L$ (preserving target identities).
\end{theorem}

\begin{proof}
For $(\Rightarrow)$: Suppose $t_1 \sim t_2$. By soundness, $t_1 \sim \nf{t_1}$ and $t_2 \sim \nf{t_2}$; combined with $t_1 \sim t_2$ by transitivity and symmetry of $\sim$, we get $\nf{t_1} \sim \nf{t_2}$, hence $\dist{\nf{t_1}} = \dist{\nf{t_2}}$. Let $\nf{t} = \mathtt{error}(p_1)\,v_1 \;;\; \cdots \;;\; \mathtt{error}(p_n)\,v_n$ where each $v_i \in \{0,1\}^{|\mathbb{D}|+|\mathbb{O}|}$ is uniquely distinct. The Walsh-Hadamard transform~\cite{macwilliams1977} of $\mu_{\nf{t}}$ evaluated at a character $\chi_s$ for $s \in \{0,1\}^{|\mathbb{D}|+|\mathbb{O}|}$ is given by $\prod_i (1-2p_i)^{s \cdot v_i}$, where $\cdot$ denotes the inner product over $\mathbb{F}_2$. Because $p_i < 0.5$ strictly implies $1-2p_i > 0$, we can take the logarithm of the transform. Under this logarithm, the coefficients corresponding to orthogonal characters $(-1)^{s \cdot v_i}$ on the Boolean hypercube $(\mathbb{Z}/2\mathbb{Z})^{|\mathbb{D}|+|\mathbb{O}|}$ are $c_{v_i} = -\frac{1}{2}\log(1-2p_i) > 0$. By orthogonality of characters, these coefficients and their associated vectors $v_i$ are uniquely and injectively determined. Thus $\dist{\nf{t_1}} = \dist{\nf{t_2}}$ uniquely identifies the set of $(p_i, v_i)$ pairs. Therefore, $G(\nf{t_1}) \cong G(\nf{t_2})$ via a target-preserving isomorphism. 
For $(\Leftarrow)$: If $G(\nf{t_1}) \cong G(\nf{t_2})$ via an isomorphism preserving target identities, then both normal forms consist of the exact same set of target multisets with the exact same associated probabilities. Therefore $\nf{t_1}$ can be transformed into $\nf{t_2}$ using only the commutativity rule R4, meaning $t_1 \sim t_2$.
\end{proof}

A direct consequence of the normal form theorem is that the normal form of a term can be computed efficiently (see Algorithm~\ref{alg:norm}), by simply applying the rewriting rules until none of them applies anymore. And since normal forms are unique, scope-local equivalence can be decided by comparing normal forms modulo same-scope permutation.

\begin{corollary}[Quasilinear decision procedure]
\label{cor:decision}
Given two terms $t_1, t_2 \in \mathbf{DEM}$, one can decide whether $t_1 \sim t_2$ by computing $\nf{t_1}$ and $\nf{t_2}$, and checking whether they are equal modulo same-scope permutation. This decision procedure is independent of the size of the repeat blocks, and runs in time $O(k|E| \log |E|)$, where $E$ is the set of instructions in $t_1$ and $t_2$ and $k$ is the maximum number of targets in an instruction.
\end{corollary}

\begin{proof}
Existence and uniqueness of normal forms are given by Theorem~\ref{thm:normalform}. Normalisation at each scope reduces to sorting instructions by target multiset and fusing adjacent same-target pairs (applying R1 and R3) in a single pass, costing $O(k|E_{\mathcal{S}}|\log|E_{\mathcal{S}}|)$ per scope $\mathcal{S}$ (where $E_{\mathcal{S}} \subseteq E$ denotes the instructions at scope $\mathcal{S}$ and $k$ bounds the size of a target set); summing over all scopes gives $O(k|E|\log|E|)$. Comparing two normal forms modulo same-scope permutation is then a sorted equality check at no additional asymptotic cost. Correctness of the procedure described by Algorithm~\ref{alg:norm} (i.e. $\nf{t} = \textsc{Normalise}(t)$) follows from the normal form theorem and the fact that the algorithm applies the rewriting rules until no rule applies anymore.
\end{proof}

In plain terms, to check if two DEMs are equivalent, we have shown that it suffices to fuse then sort their error events by their effective target sets:

\begin{enumerate}
  \item cancel any detector or observable that appears an even number of times in the same error event.
  \item simplify each DEM by combining same-target error events.
  \item sort the remaining error events by their target sets.
  \item compare the sorted lists, so that if they match, the DEMs are equivalent.
\end{enumerate}

In practice, standard QEC schemes (e.g. surface codes) have physical errors that trigger a small bounded number of detectors (e.g., $k \le 4$), making $k$ a negligibly small constant. Concretely, this decision procedure drastically improves the efficiency of equivalence checking when compared to a naive sampling-based approach.
Moreover, when many instructions share the same target sets (which is fairly common in practice), implementing Algorithm~\ref{alg:norm} with a hash map to group same-target instructions can further reduce the (average) time complexity to $O(k|E|)$, since the number of distinct target multisets is often negligible compared to the total number of instructions.

\begin{algorithm}[h]
\caption{DEM Term Normalisation Procedure}\label{alg:norm}
\begin{algorithmic}[1]
\Function{Normalise}{$t$}
    \State Let $E_{\text{flat}}$ be the sequence of flat instructions in $t$'s top-level scope.
    \State Apply R3: for each instruction $\mathtt{error}(p) S$ in $E_{\text{flat}}$, reduce $S$ to its effective target set $\widehat{S}$ (retaining only targets appearing an odd number of times).
    \State Apply R2: remove instructions from $E_{\text{flat}}$ where $p=0$ or $\widehat{S}=\emptyset$.
    \State Sort the remaining instructions in $E_{\text{flat}}$ lexicographically by their reduced target sets $\widehat{S}$.
    \State $L \gets []$
    \For{each instruction $\mathtt{error}(p) S$ (in the sorted list $E_{\text{flat}}$)}
        \If{$L$ is not empty and the last instruction in $L$ has target set $S$}
            \State Let $\mathtt{error}(p_{\text{last}}) S$ be the last instruction in $L$
            \State $p_{\text{new}} \gets p_{\text{last}} \oplus p$ \Comment{Rule R1 Fusion}
            \State Update last instruction in $L$ to $\mathtt{error}(p_{\text{new}}) S$
        \Else
            \State Append $\mathtt{error}(p) S$ to $L$
        \EndIf
    \EndFor
    \For{each \texttt{REPEAT} block $b$ at the top-level of $t$}
        \State Replace the body of $b$ with \Call{Normalise}{body of $b$}
    \EndFor
    \State \Return the term consisting of the sequentially composed instructions in $L$ followed by the normalised \texttt{REPEAT} blocks.
\EndFunction
\end{algorithmic}
\end{algorithm}

\section{Completeness and adaptive QEC}
\label{sec:completeness}

We now discuss the relationship between scope-local equivalence and decoder-equivalence. Decoder-equivalence is a semantic notion of equivalence for DEMs: two DEM terms are decoder-equivalent if they induce the same joint probability distribution over detector outcomes and observable outcomes (Definition~\ref{def:decoder-equivalence}). Scope-local equivalence is a semantic notion; the rewriting system $\R$ provides its syntactic characterisation. Scope-local equivalence implies decoder-equivalence: if $t_1 \sim t_2$, then $\dist{t_1|_{\mathcal{S}}} = \dist{t_2|_{\mathcal{S}}}$ for every scope $\mathcal{S}$, and since both terms share the same skeleton, their unrolled global distributions are equal. The converse direction would imply that DEM term normal forms characterise decoder-equivalence, and thus that the decision procedure given by Corollary~\ref{cor:decision} decides decoder-equivalence.

To the best of our knowledge, there is no standard QEC pipeline that generates nested repeat blocks. Indeed, syndrome extraction is commonly performed with a single periodic cycle. One could argue that ``meta-cycles'' requiring nested repeat blocks do not occur naturally for QEC schemes that do not involve adaptive circuits with classical feedback loops.

In some QEC pipelines, logical measurements occur at fixed points outside repeat blocks, so that each observable index appears in the instructions of at most one scope. The following definition captures this property.

\begin{definition}[Observable-separated term]
\label{def:obs-separated}
A DEM term $t$ is \emph{observable-separated} if every observable index $L_j \in \mathbb{O}$ appears in the instructions of at most one scope of $t$.
\end{definition}

In Stim, since detector indices are already scope-disjoint by Stim's shifting convention, the global unrolled distribution $\dist{\unroll(t)}$ of a term $t$ is the iterated convolution of its scope distributions (for a term with a single \texttt{REPEAT} block of count $N$, so that $\dist{\unroll(t)} = \dist{t_{\text{top}}} \circledast (\dist{t_{\text{body}}})^{\circledast N}$), where $\dist{}^{\circledast N} = \dist{} \circledast \cdots \circledast \dist{}$, $N$ times. Observable-separated terms are precisely those whose per-scope error events are jointly independent. This allows to prove completeness for pairs of terms that are both observable-separated and share the same observable-scope allocation.

\begin{proposition}[Completeness for observable-separated, detector-separated terms]
\label{prop:completeness}
Let $t_1, t_2$ be skeleton-compatible DEM terms such that:
\begin{enumerate}
  \item Both are observable-separated (each observable appears in at most one scope).
  \item Both have the same observable-scope allocation.
  \item Both are detector-separated (each detector appears in at most one scope).
\end{enumerate}

Then: $t_1 \sim t_2$ (scope-locally equivalent) if and only if 
$t_1 \approx t_2$ (decoder-equivalent).
\end{proposition}

\begin{proof}
For $(\Rightarrow)$: if $t_1 \sim t_2$, then $\dist{t_1|_{\mathcal{S}}} = \dist{t_2|_{\mathcal{S}}}$ for every scope $\mathcal{S}$. Both terms are observable-separated, so each observable appears in at most one scope. Since both terms share the same skeleton, their unrolled global distributions are products of independent scope distributions. By equality of per-scope distributions, the global distributions coincide:
$\dist{\unroll(t_1)} = \dist{\unroll(t_2)}$, so $t_1 \approx t_2$.

For $(\Leftarrow)$: suppose $t_1 \approx t_2$, i.e., $\dist{\unroll(t_1)} = \dist{\unroll(t_2)}$. Since both terms are observable-separated with the same observable-scope allocation and share the same skeleton, their unrolled distributions factor identically as independent products over scopes. The factorisation is injective: equal global distributions imply equal per-scope distributions at every scope. Hence $\dist{t_1|_{\mathcal{S}}} = \dist{t_2|_{\mathcal{S}}}$ for all scopes, so $t_1 \sim t_2$.
\end{proof}

\begin{remark}[Detector-separation in Stim-generated DEMs]
By Stim's \texttt{shift\_detectors} convention, all DEMs generated by Stim are 
detector-separated. Thus, Proposition~\ref{prop:completeness} 
applies to all standard QEC pipelines in Stim without additional preconditions.
\end{remark}

It follows that Corollary~\ref{cor:decision} decides \emph{full} decoder-equivalence whenever the two DEMs being compared are observable-separated with the same scope allocation. In particular, for \textbf{all} DEMs generated by standard QEC pipelines in Stim, where observables appear exclusively in the top-level scope. This is the practical scope for non-adaptive QEC: in standard quantum error correction, logical measurements occur at fixed points (typically between syndrome extraction cycles), ensuring observable-separated DEM generation. Thus, the $O(k|E|\log|E|)$ decision procedure settles full decoder-equivalence for all non-adaptive QEC implementations (e.g., surface code, colour code, or Floquet code memory experiments) regardless of circuit depth $N$ or repeat-block nesting depth, solving what statistical equivalence testing cannot solve efficiently at all.

Complex lattice surgery pipelines with adaptive corrections (where intermediate observable measurements inform subsequent operations) may violate observable-separation if observables are measured at different cycle boundaries; for such pipelines, the completeness guarantee does not apply and the procedure decides scope-local equivalence, which implies but does not completely characterise decoder-equivalence in that setting. 

The matching-allocation hypothesis in Proposition~\ref{prop:completeness} cannot be dropped, even when both terms are individually observable-separated. Consider
\[
  t_1 = \mathtt{REPEAT}\;2\;\{\mathtt{error}(0.1)\;L_0\},
  \qquad
  t_2 = \mathtt{REPEAT}\;2\;\{\varepsilon\}\;;\;\mathtt{error}(0.18)\;L_0.
\]
Both have skeleton $[\mathrm{node}(2,[])]$ and are observable-separated: $L_0$ lies in exactly one scope of each term, but in different scopes---the body scope of $t_1$ versus the top-level scope of $t_2$. Unrolling yields $P(L_0=1) = 0.1\oplus 0.1 = 0.18$ for $t_1$ and $P(L_0=1) = 0.18$ for $t_2$ directly, so they are decoder-equivalent. However, the body-scope distributions differ ($\bern(0.1)$ vs.\ $\delta_0$), so $t_1\not\sim t_2$. The general relationship between decoder-equivalence and scope-local equivalence for terms with mismatched observable-scope allocations remains open.

Proposition~\ref{prop:completeness} relies on the \emph{observable-separated} property. As noted in Section~\ref{sec:terms}, this property holds in standard QEC pipelines, where coordinate shifts (such as Stim's \texttt{shift\_detectors}) can be statically resolved into absolute, disjoint indices. However, this static resolution imposes a rigid schedule for logical measurements. Modern QEC architectures increasingly rely on adaptive protocols where mid-protocol measurements dictate subsequent operations. To minimise hardware latency, these conditional operations are typically tracked entirely in software rather than actively applied to the quantum state. The control system maintains a \emph{Pauli frame} (a classical boolean registry of pending Pauli corrections), and uses it to invert the signs of subsequent measurement outcomes whenever their observables would have been flipped by the pending corrections. This adaptivity is central to several techniques:

\begin{itemize}
    \item \textbf{Lattice Surgery:} Merging and splitting logical patches requires measuring joint boundary parities~\cite{horsman2012latticesurgery,litinski2019lattice}. Random measurement outcomes force adaptive Pauli frame updates in subsequent steps.
    \item \textbf{Distributed QEC:} In modular quantum architectures, logical states and entanglement are routed via teleportation interfaces. The Bell measurement outcomes generated between network nodes require Pauli feedback corrections on the remote data qubits, forcing adaptive Pauli frame updates to track logical observables~\cite{xuetal2022,sutcliffe2026,shalby2025,jacinto2026}.
\end{itemize}

In all these cases, a measurement in scope $k$ conditionally influences the events and operations in scope $k+1$. At the level of the quantum circuit, this is a feed-forward process. 
However, DEMs have no classical control flow (if-then-else). Instead, the adaptivity manifests structurally: when Stim's extraction functor propagates errors through adaptive Pauli frame updates, error mechanisms in different scopes may both flip the same logical observable~\cite{gidney2021stim}. Because the logical observable index leaks across structural scope boundaries, it appears in the \texttt{error} instructions of multiple distinct scopes. 
Consequently, the strict \emph{observable-separated} condition is violated.
To recover completeness for partially-adaptive circuits, we introduce the concept of \emph{observable-coherence classes}:

\begin{definition}[Observable-coherence classes]
Two scopes in a DEM term belong to the same \emph{observable-coherence} class if they share a logical observable in their instructions, extended by transitive closure. A coherence class corresponds to a connected component in the bipartite graph where the left vertices represent logical observables, the right vertices represent structural scopes, and an edge connects an observable to a scope if that observable appears in at least one instruction of that scope.
\end{definition}

\begin{remark}
Observable-coherence classes are defined structurally on \emph{scopes}, not on individual error mechanisms. The underlying graph is the scope-observable bipartite graph, distinct from the Tanner graph used elsewhere (which represents targets and error mechanisms). This partition allows the global probability distribution to factor cleanly: disjoint classes correspond to independent sets of error mechanisms with no shared observables.
\end{remark}

Standard QEC pipelines typically ensure that detector indices are not shared across scopes. In Stim, \texttt{shift\_detectors} instructions are used in repeat blocks to ensure that each iteration has its own disjoint set of detector indices. Resolving these shifts (as assumed in Definition~\ref{def:grammar}) yields a DEM term which is \emph{detector-separated}: detector indices are scope-disjoint. 

Instead of requiring the entire term to be strictly observable-separated, we partition the scopes based on these classes. We establish the following two compositional properties:

\begin{lemma}[Marginal observable distribution factors over observable-coherence classes]
\label{lem:marginal-factorization}
The \emph{marginal distribution over observables (only)} of a DEM term $t$ factors over disjoint observable-coherence classes:

If $C_1, C_2$ are disjoint observable-coherence classes with observable sets $\mathcal{O}_1 \subseteq \mathbb{O}$ and $\mathcal{O}_2 \subseteq \mathbb{O}$ (disjoint), then:
\[
  P(L_{C_1}, L_{C_2}) = P(L_{C_1}) \otimes P(L_{C_2})
\]
where $P(L_i)$ denotes the marginal distribution on observables in class $C_i$.
\end{lemma}

\begin{proof}
Observables in disjoint coherence classes depend on disjoint sets of error mechanisms. Since error mechanisms are independent across scopes (by the structure of DEM unrolling), the observable outcomes are independent across disjoint classes.
\end{proof}

\begin{lemma}[Joint distribution factors for detector-separated terms]
\label{lem:joint-factorization}
For a detector-separated DEM term $t$, the \emph{full joint distribution over detectors and observables} factors over disjoint observable-coherence classes:

If $C_1, C_2$ are disjoint observable-coherence classes with index sets $I_1, I_2 \subseteq \mathbb{D} \cup \mathbb{O}$ (disjoint), then:
\[
  P(D_{I_1}, L_{I_1}, D_{I_2}, L_{I_2}) = P(D_{I_1}, L_{I_1}) \otimes P(D_{I_2}, L_{I_2})
\]
where $P(I)$ denotes the marginal on the index set $I$.
\end{lemma}

Note that detector-separation is necessary in this lemma: without detector-separation, a detector may appear in multiple coherence classes, coupling them and preventing factorization.

\begin{proof}
By detector-separation, each detector appears in at most one scope, hence in at most one observable-coherence class. Thus, the index sets $I_1, I_2$ partition both the detector and observable indices into disjoint, independent components. The full joint distribution factors accordingly.
\end{proof}

This enables a hybrid decision procedure. For the static, observable-separated majority of the circuit (where a coherence class is confined to a single scope), equivalence is decided via the $O(k|E|\log|E|)$-time algebraic rewriting framework without any unrolling. For the isolated adaptive segments where a class spans multiple structural scopes, the procedure performs a strictly local unrolling. By flattening only the overlapping coherent classes, we trivially recover the separation property locally for that segment.

\begin{theorem}[Soundness of scope-local equivalence for decoder-equivalence]
\label{thm:generalised-soundness}
Let $t_1, t_2 \in \mathbf{DEM}$ be skeleton-compatible, detector-separated terms where each observable index appears in the same set of observable-coherence classes in both terms. If $t_1 \sim t_2$, then $t_1 \approx t_2$.
\end{theorem}

\begin{proof}
By Lemma~\ref{lem:joint-factorization}, since both terms are detector-separated, their joint distributions factor over observable-coherence classes. Since $t_1$ and $t_2$ have the same skeleton and identical scope distributions at every scope (by hypothesis), their global distributions factor identically. For each class $\mathcal{C}_i$, the marginal distribution $\dist{t_1|_{\mathcal{C}_i}}$ depends only on:
\begin{enumerate}
  \item The scope distributions $\dist{t_1|_{\mathcal{S}}}$ for all scopes $\mathcal{S} \in \mathcal{C}_i$.
  \item The repeat structure (skeleton) within the tree of scopes forming $\mathcal{C}_i$.
\end{enumerate}
Since both terms are skeleton-compatible and have identical scope distributions, we have $\dist{t_1|_{\mathcal{C}_i}} = \dist{t_2|_{\mathcal{C}_i}}$ for every class $\mathcal{C}_i$.

Therefore:
\[
  \dist{t_1|_E} = \circledast_{i=1}^m \dist{t_1|_{\mathcal{C}_i}} = \circledast_{i=1}^m \dist{t_2|_{\mathcal{C}_i}} = \dist{t_2|_E},
\]
establishing that $t_1$ and $t_2$ are decoder-equivalent.
\end{proof}

This hybrid approach provides a sound and efficient decision procedure for adaptive protocols, without suffering the exponential state-space explosion of globally unrolling the term.

\section{Discussion}
\label{sec:discussion}

The present work provides the first static characterisation of DEM equivalence. Moving beyond simulation-based heuristics is essential to improve the mathematical reliability of quantum compilation pipelines. Crucially, the decision procedure is complete for non-adaptive pipelines, and provides a sound and efficient hybrid certification for adaptive systems: it settles full decoder-equivalence exactly for standard pipelines in $O(k|E|\log|E|)$ time complexity, and restricts unrolling repeat blocks to overlapping features in adaptive systems. Static checks for equivalences can be used to verify compiler passes, and to optimise DEMs while preserving their semantics. It is also an essential tool for compiler regression testing: for example, checking that distinct versions of a compiler generate semantics-equivalent DEMs.

While a completeness result for adaptive protocols would be a satisfying theoretical milestone, it only marginally improves the practical utility of our approach, since the hybrid procedure already provides a sound and efficient decision procedure for adaptive protocols. The lack of such completeness result means that the decision procedure might produce false negatives for some adaptive protocols (like magic state cultivation~\cite{gidney2024magicstatecultivation}). But in practice, the false negative rate is likely negligible for compiler regression testing (as the DEMs compared come from the same compiler pipeline). A more practical direction is to explore how the algebraic structure of DEMs can be exploited to design new compiler optimisations that are provably semantics-preserving by construction. The present work validates that Stim handles DEMs in a way that is very consistent with their algebraic structure, and in doing so, we outlined good practice in handling DEMs in Stim-like libraries (see e.g.~\cite{chase2026clifftfastexactsimulation, fang2026lightstimframeworkqecprotocol}).

It would be interesting for example to explore how the algebraic structure of DEMs relates to the ZX-calculus~\cite{peham2022equivalenceZX, carette2019szx}, connecting our work to the broader field of diagrammatic languages for quantum computing. It would also be valuable to understand the underlying denotational semantics of DEMs, and possibly augment the \texttt{.dem} format with inductive types as it was done with quantum circuits (see e.g.~\cite{pechoux2020, rennelastaton2020}). The use of the Giry Monad for error mechanisms echoes effect handlers in probabilistic programming~\cite{plotkin2013handlingalgebraiceffects}, opening the door to a study of QEC compilation pipelines as algebraic effects. Moreover, the semantic invariants in our DEM term language can be encoded as refinement types, enabling the use of proof assistants (Rocq, Lean) to automate the verification of compiler passes.

Finally, the original motivation for this work is the development of \texttt{emlint}~\cite{emlint}, an open-source library for static analysis of detector error models. A natural next step is to integrate the decision procedure of Corollary~\ref{cor:decision} into \texttt{emlint}, to assert that compiler optimisation passes leave the semantics of DEMs unchanged.

\paragraph*{Acknowledgements}
This work has been supported by the Mozilla Foundation. I am grateful to the Stim community and the tqec community for their engagement and feedback on early versions of \texttt{emlint}. AI assistance (Claude Haiku 4.5, Claude Sonnet 4.6, Gemini 3.1 Pro) was used for copy-editing, proofreading, and LaTeX formatting.

\bibliographystyle{plain}
\bibliography{refs}

\end{document}